\documentclass[a4paper,10pt]{amsproc}

\usepackage{amssymb,amsthm}
\usepackage{cite}
\usepackage{url}

\theoremstyle{plain}
\newtheorem{theorem}{Theorem}

\newtheorem{proposition}{Proposition}
\newtheorem{corollary}{Corollary}
\theoremstyle{definition}
\newtheorem{definition}{Definition}

\newtheorem{condition}{Condition}

\newcommand{\bn}{\mathbf{n}}
\def\RR{\mathbf{R}}
\def\lra{\longrightarrow}
\def\a{\alpha}
\def\ga{\gamma}
\def\la{\lambda}
\def\e{\varepsilon}
\def\si{\sigma}
\def\ph{\varphi}
\def\om{\omega}
\def\Ga{\Gamma}
\def\La{\Lambda}
\DeclareMathOperator\rank{rank}
\DeclareMathOperator\const{const}
\DeclareMathOperator\tr{tr}
\DeclareMathOperator*{\var}{var}
\DeclareMathOperator{\ind}{ind}
\newcommand{\I}{{\overline{I}}}
\def\cF{{\mathcal{F}}}
\def\cI{{\mathcal{I}}}
\def\cJ{{\mathcal{J}}}
\def\ov{\overline}
\def\wt{\widetilde}
\def\wh{\widehat}
\def\pa{\partial}
\newcommand\pd[2]{\frac{\pa #1}{\pa #2}}
\let\rom\textup
\newcommand{\BL}{\biggl}
\newcommand{\BR}{\biggr}
\newcommand{\Bl}{\Bigl}
\newcommand{\Br}{\Bigr}
\newcommand{\bl}{\bigl}
\newcommand{\br}{\bigr}
\newcommand{\abs}[1]{\lvert#1\rvert}
\newcommand{\1}[1]{\overset1{#1}}

\begin{document}

\author{S.~Yu.~Dobrokhotov}
\address{A.~Ishlinsky Institute for Problems in Mechanics,
Moscow; Moscow Institute of Physics and Technology, Dolgoprudny,
Moscow District}
\email{dobr@ipmnet.ru}
\thanks{Supported by RFBR (grant~11-01-00973-a) and by the
Archimedes Center for Modeling, Analysis \& Computation (ACMAC),
Crete, Greece (grant FP7-REGPDT-2009-1). S.~Yu.~D. and V.~E.~N. are
grateful to the staff of ACMAC and the Department of Applied
Mathematics, University of Crete for support and kind hospitality.}

\author{G.~Makrakis}
\address{Department of Applied
Mathematics, University of Crete; Institute of Applied \&
Computational Mathematics, Foundation for Research and
Technology-Hellas, Heraklion, Crete, Greece}
\email{makrakg@iacm.forth.gr}

\author{V.~E.~Nazaikinskii}
\address{A.~Ishlinsky Institute for Problems in Mechanics,
Moscow; Moscow Institute of Physics and Technology, Dolgoprudny,
Moscow District}
\email{nazay@ipmnet.ru}

\title[A new  representation of Maslov's canonical operator]%
{Fourier integrals and a new representation\\ of Maslov's canonical
operator near caustics}

\dedicatory{Dedicated to Vladimir Aleksandrovich Marchenko}

\begin{abstract}
We suggest a new representation of Maslov's canonical operator in a
neighborhood of the caustics using a special class of coordinate
systems (``eikonal coordinates'') on Lagrangian manifolds.
\end{abstract}


\maketitle

\setcounter{tocdepth}{2}

\section*{Introduction}

Rapidly oscillating Fourier type integrals are well known in the mathematical literature. One of the main constructions in this area is given by
Maslov's canonical operator \cite{Mas1} (see also \cite{MaFe1,MSS1}
and the bibliography therein), which is used to construct short-wave
(high-frequency, or rapidly oscillating) asymptotic solutions of a
broad class of problems for differential equations with real
characteristics.\footnote{There is also a version of the canonical
operator for equations with complex characteristics (see
\cite{Mas6,Mas3} and also \cite{BelDob92,MaNa2,MSS2}), which we do
not discuss here.} The asymptotics provided by the canonical
operator are a far-reaching generalization of ray expansions in
problems of optics and electrodynamics as well as of the WKB
asymptotics for equations of quantum mechanics. The construction of
the canonical operator is based on the fundamental geometric notion
of a \textit{Lagrangian manifold}. Assume that the original partial
differential equation is defined on the $n$-dimensional
configuration space $\mathbb{R}^n_x$ with coordinates
$x=(x_1,\ldots,x_n)$. One of the main ideas in the canonical
operator is to proceed from this equation to a simpler (in fact,
ordinary) differential equation naturally induced on an
$n$-dimensional Lagrangian manifold $\Lambda$ in the phase space
$\mathbb{R}^{2n}_{(x,p)}$ with coordinates $(x,p)$,
$x=(x_1,\ldots,x_n)$. The manifold $\La$ depends on the problem
considered and is usually constructed by solving the canonical
equations of classical (Hamiltonian) mechanics. Once we have found
an appropriate manifold $\La$ and a solution $\phi$ (which is called
an \textit{amplitude}) of the induced differential equation on
$\La$ (the choice of a specific solution depends on the original
problem as well), we can write out the (asymptotic) solution of the
original problem in the form
\begin{equation}\label{KO}
u(x)=[K_{\Lambda} a](x),
\end{equation}
where $K_{\Lambda}$ is the canonical operator. Note that
\begin{itemize}
    \item There exist known recipes or algorithms for
constructing Lagrangian manifolds and amplitudes for many types of
problems (and for various original differential equations).
    \item Formula \eqref{KO} is the \textit{answer to
the original problem}, automatically including the behavior in
caustic regions, passage through the caustics, matching of the
asymptotic expansions in various regions, etc.
\end{itemize}
The standard construction \cite{Mas1} of the canonical operator is
universal, but it is not the only possible one. For a broad class
of interesting problems, one can more conveniently use different
representations, and it is these new representations of oscillating
solutions that are considered in our paper. We point out that our
analysis does not alter the general \textit{concept} of the
canonical operator or the \textit{fundamental geometric objects}
underlying its construction. We just suggest a new implementation
of the canonical operator in a neighborhood of the caustics, which
can be more convenient when solving specific physical problems with
the use of software like \textsl{Wolfram
Mathematica}$^\circledR$~\cite{Mathematica} or
\textsl{MatLab}$^\circledR$~\cite{MatLab}. Note also that our
formulas are a special case of the general formulas of the theory
of Fourier integral operators~\cite{Hor6}, and our main result is
the specific form of these formulas and an algorithm for their
construction. Finally, note that our formulas may in particular be
useful in problems related to the asymptotics of solutions of wave
type equations with localized initial data (e.g., see
\cite{DShTMZ,DShT}) or in scattering problems and point source
problems for equations of Helmholtz type.

The outline  of the paper is as follows. The standard construction
of the canonical operator is described in Section~\ref{s1}. In
Section~\ref{s2}, we describe the main result, our new formulas~\eqref{e05} and~\eqref{e06}.
Section~\ref{s3} provides some examples. Section~\ref{s4} describes
the relationship, much discussed at a certain time in the past,
between the canonical operator and Fourier integrals. Finally,
Section~\ref{s6} contains auxiliary material: it describes some
notation used in the paper and also, for the reader's convenience,
reproduces the famous theorem of the stationary phase method. We
omit the proofs (which are mostly technical and involve lengthy
computations) everywhere except for Section~\ref{s4}, where we feel
that the short proofs provided might be of interest.

The asymptotics discussed here have the small parameter $h\to0$, as
is customary in the semiclassical approximation in quantum
mechanics. In wave problems, one often uses the large parameter
$k\to\infty$. To make the formulas fit this case, one should just
set $h=1/k$.

\section{Standard construction of the canonical operator}\label{s1}

In this section, we recall the construction of Maslov's canonical
operator according to~\cite{Mas1,MaFe1,MSS1}. Let us start from a
very brief overview and then fill in the details.

The input elements of the construction are as follows:
\begin{itemize}
    \item A \textit{Lagrangian manifold}
$\Lambda$ in the $2n$-dimensional phase space $\RR_{(x,p)}^{2n}$.
    \item A \textit{measure} $d \mu$ on $\Lambda$.
    \item A point $\alpha_0\in\Lambda$, referred to as the
    \textit{central point}.
\end{itemize}
If the \textit{quantization conditions} are satisfied, then these
elements uniquely (modulo lower-order terms) determine the
canonical operator $K_{(\La,d\mu)}^{1/h}$, which takes each
function $ a\in C_0^\infty(\La)$ to a rapidly oscillating function
$u(x,h)= [K_{(\La,d\mu)}^{1/h} a](x)$ on the configuration space
$\RR_x^n$. The construction is essentially \textit{local}: first,
$[K_{(\La,d\mu)}^{1/h} a](x)$ is defined for functions $ a$
supported in certain open sets called the \textit{canonical charts}
on $\La$; then a partition of unity on $\La$ is used to paste the
local definitions together into the \textit{global canonical
operator}. The local expression for $[K_{(\La,d\mu)}^{1/h} a](x)$
in a canonical chart has the form of a rapidly oscillating
exponential in the simplest case where the chart is
diffeomorphically projected onto a domain in the configuration
space $\RR_x^n$ (a \textit{nonsingular chart}); in a
\textit{singular chart} (a chart containing a \textit{focal
point}), the local expression has the form of the Fourier transform
of a rapidly oscillating exponential with respect to part of the
variables. A change of the central point $\alpha_0$ results in the
multiplication of all local expressions by a unimodular phase
factor. In general, the local expressions depend not only on the
position of the central point but also on the choice of paths from
the central point to the respective canonical charts. The role of
the quantization condition is that it guarantees that the local
expressions for the canonical operator are independent of the
choice of these paths and coincide with each other on functions
supported in intersections of canonical charts.

Note that although the canonical operator is an object of function
theory significant in its own right, its main applications are
related to various problems for partial differential equations,
hence the importance of the \textit{commutation formula}, which
shows how a differential operator acts on the function
$[K_{(\La,d\mu)}^{1/h} a](x)$ and provides conditions (in the form
of a geometric condition on $\La$ and an ordinary differential
equation for $ a$) ensuring that this function is a solution of
the corresponding differential equation.

Now let us proceed to more detailed explanations.

\subsection{Lagrangian manifold, measure, and central point}

First, let us discuss the input elements of the construction. A
\textit{Lagrangian manifold} in $\RR_{(x,p)}^{2n}$ is an
$n$-dimensional submanifold $\La\subset\RR_{(x,p)}^{2n}$ such that
the symplectic form $\om^2=dp\wedge dx\equiv dp_1\wedge
dx_1+\dotsm+dp_n\wedge dx_n$ vanishes on the vectors tangent to
$\La$. We denote the points of $\La$ by the letter $\a$ and use the
notation $\a=(\a_1,\dotsc,\a_n)$ for various local coordinate
systems on $\La$. Then the embedding $\La\subset\RR_{(x,p)}^{2n}$
is given by equations of the form $x=X(\a)$, $p=P(\a)$, $\a\in\La$.

Next, a \textit{measure} on $\La$ is understood as a volume
form\footnote{Thus, $\La$ is orientable; the theory may pretty well
be constructed without this assumption, which we only make to
simplify the exposition by avoiding the notion of \textit{odd}
differential forms.} (a nonvanishing differential $n$-form) $d\mu$.
In local coordinates $(\a_1,\dotsc,\a_n)$, one has $d\mu=\mu
(\a)\,d\a_1\wedge\dotsm\wedge d\a_n$, where the function $\mu
(\a)\ne0$ is called the \textit{density} of $d\mu$ in these
coordinates.

Finally, a \textit{central point} is an arbitrarily chosen point
$\a_0\in\La$. (We assume $\La$ to be connected; otherwise, we need
one central point per connected component.)

\subsection{Regular and focal points. Canonical
coordinates on $\La$}

Let $\a_*\in\La$. If $\det\pd X\a(\a_*)\ne0$ (this condition is
independent of the choice of local coordinates), then the point
$\a_*$ is said to be \textit{regular}; otherwise, it is said to be
\textit{singular}, or \textit{focal}. If $\a_*$ is a regular point,
then the equation $x=X(\a)$ can be solved for $\a$ in a
neighborhood of the point $(x_*,\a_*)$, $x_*=X(\a_*)$, and hence
the variables $x=(x_1,\dotsc,x_n)$ can be used as local coordinates
on $\La$ in a neighborhood of~$\a_*$. For an arbitrary
$\a_*\in\La$, the \textit{lemma on local coordinates}~\cite{Arnold}
states that there exists a subset $I\subset\{1,\dots,n\}$ such that
$\det\pd{(X_I,P_\I)}\a(\a_*)\ne0$ (where $\I=\{1,\dots,n\}\setminus
I$ is the complementary subset); consequently, the equations
$x_I=X_I(\a)$, $p_\I=P_\I(\a)$ can be solved for $\a$ in a
neighborhood of the point $(x_{I*},p_{\I*},\a_*)$,
$x_{I*}=X_I(\a_*)$, $p_{\I*}=P_\I(\a_*)$, and the variables
$(x_I,p_\I)$ can be used as local coordinates on $\La$ in a
neighborhood of~$\a_*$. These coordinates are called
\textit{canonical coordinates}. For a regular point $\a_*$, we can
of course take $\I=\varnothing$, so that $(x_I,p_\I)=x$, but if
$\a_*$ is a focal point, then $\I$ is necessarily nonempty.

It follows from the preceding that there exists a \textit{canonical
atlas} of $\La$ in which every chart is given for some
$I\subset\{1,\dots,n\}$ by the canonical coordinates $(x_I,p_\I)$
defined on an open connected simply connected subset $U\subset\La$;
such a chart is denoted by $(U,I)$ and called a \textit{canonical
chart} (\textit{nonsingular} if $\I=\varnothing$ and
\textit{singular} otherwise). The equations of $\La$ in the
canonical chart have the form
\begin{equation}\label{loc-coor0}
    x_\I=X_\I(x_I,p_\I),\quad p_I=P_I(x_I,p_\I),
\end{equation}
(where for brevity we write $X_\I(x_I,p_\I)$ instead of
$X_\I(\a(x_I,p_\I))$ etc.).

\subsection{Maslov index of paths and cycles on $\La$}

Let $\e\ge0$ be a given number. The form $d(X_1-i\e
P_1)\wedge\dotsm\wedge d(X_n-i\e P_n)$ is a differential form of
maximal degree $n$ on $\La$ and hence a multiple of $d\mu$; thus,
the function
\begin{equation}\label{Je}
\cJ^{\e}(\a)=\frac{d(X_1-i\e P_1)\wedge\dotsm\wedge d(X_n-i\e
P_n)}{d\mu}
\end{equation}
is well defined; it is called the \textit{Jacobian} of the
functions $(X_1-i\e P_1,\dotsc, X_n-i\e P_n)$ \textit{with respect
to the measure} $d\mu$ and can be computed by the formula
\begin{equation}\label{Je1}
\cJ^{\e}(\a)=\frac{1}{\mu(\a)}
  \det\pd{(X_1-i\e P_1,\dotsc,X_n-i\e P_n) }{(\a_1,\dotsc,\a_n)},
\end{equation}
where $\mu (\a)$ is the density of $d\mu$ in local coordinates
$(\a_1,\dotsc,\a_n)$. It can be shown that \textit{for $\e>0$ the
Jacobian \eqref{Je} vanishes nowhere on $\La$}. On the other hand,
for $\e=0$ the Jacobian~\eqref{Je} becomes the Jacobian
\begin{equation}\label{Jnons}
  \mathcal{J}(\a)=\frac{dX_1\wedge\dotsm\wedge dX_n}{d\mu}
  \equiv\frac{1}{\mu(\a)}\det\pd{(X_1,\dots,X_n)}
  {(\a_1,\dotsc,\a_n)},
\end{equation}
which is nonzero at the regular points and vanishes at the focal
points.

Let $\ga\colon[0,1]\lra\La$ be a path on $\La$ with regular
endpoints $\ga(0)$ and $\ga(1)$. The \textit{Maslov index} of $\ga$
is defined by\footnote{See~\cite{Arnold} for the definition of the
Maslov index of $\ga$ as the intersection number of $\ga$ with the
cycle of singularities on $\La$.}
\begin{equation}\label{MasInd}
\ind\ga=\frac{1}{\pi}\lim_{\varepsilon\to+0}\var_\ga\arg \cJ^\e(\a)
=\frac{1}{\pi i}\lim_{\varepsilon\to+0} \int_\ga
\frac{d\cJ^\e}{\cJ^\e},
\end{equation}
where $\arg z$ is the argument of a complex number $z$ and
$\var_\ga$ stands for the variation along $\ga$. The index
$\ind\ga$ is an integer depending only on the homotopy class of
$\ga$ in the set of paths with regular endpoints. Next, if $\ga$ is
a closed path (\textit{cycle}) on $\La$, then the Maslov index of
$\ga$ is defined by the same formula~\eqref{MasInd} (the limit as
$\e\to+0$ being in fact unnecessary, because the variation of the
argument is independent of $\e$ in this case). The Maslov index of
a cycle is an integer homotopy invariant depending only on the
homology class of the cycle in $H_1(\La)$.

\subsection{Jacobians and Maslov index of canonical charts}\label{s14}

From now on, to avoid unnecessary technical complications, we
assume that \textit{the central point $\a_0$ is nonsingular and
$\cJ(\a_0)>0$}. Let $(U,I)$ be some canonical chart. Then the
Jacobian
\begin{equation}\label{JI}
    \cJ_I(\a)
    =\frac{dX_I\wedge dP_\I}{d\mu}
  \equiv\frac{1}{\mu(\a)}\det\pd{(X_I,P_\I)}
  {(\a_1,\dotsc,\a_n)}
\end{equation}
is nonzero in $U$. (Note that for $\I=\varnothing$ this is just the
Jacobian~\eqref{Jnons}.) The Maslov index $m_{(U,I)}$ of the chart
$(U,I)$ is defined as follows. Choose some path
$\ga\colon[0,1]\lra\La$ with $\ga(0)=\a_0$ and $\ga(1)\equiv\a_*\in
U$. Assume momentarily that $\a_*$ is a nonsingular point. Then
\begin{equation}\label{MasIndCh1}
    m_{(U,I)}=\ind\ga+\frac1\pi
    \left[\arg \frac{dX_I\wedge d((1-\theta)X_I
    -i\theta P_\I)}{d\mu}(\a_*)\right]\bigg|_{\theta=0}^{\theta=1}
    +\frac{\abs{\I}}2,
\end{equation}
where $[\arg(\cdot)]_{\theta=0}^{\theta=1}$ is the variation of the
argument as $\theta$ varies from $0$ to $1$ on the interval
$[0,1]$. If $\a_*$ cannot be assumed to be nonsingular, then one
can use the slightly more cumbersome formula
\begin{equation}\label{MasIndCh2}
    m_{(U,I)}=\frac1\pi
    \var_\ga\arg \frac{d(X_I-i\e(t)P_I)\wedge d(\theta(t)X_I
    -i\varkappa(t) P_\I)}{d\mu}(\a_*)
    +\frac{\abs{\I}}2,
\end{equation}
where $e(t)$, $\theta(t)$, and $\varkappa(t)$ are continuous
functions on $[0,1]$ positive on $(0,1)$ and satisfying the
conditions $\theta(0)=\varkappa(1)=1$ and $\varkappa(0)=
\e(0)=\e(1)=\theta(1)=0$. Note that (i) formulas \eqref{MasIndCh1}
and \eqref{MasIndCh2} agree if $\a_*$ is a nonsingular point; (ii)
$m_{(U,I)}$ is an integer, and $\pi m_{(U,I)}$ is a branch of
$\arg\cJ_I$ in $U$; (iii) $m_{(U,I)}=\ind\ga$ if $\I=\varnothing$
(i.e., $(U,I)$ is a nonsingular chart); (iv) $m_{(U,I)}$
\textit{depends on the choice of the \rom(homotopy class of
the\rom) path} $\ga$ (unless the Maslov index of all cycles on
$\La$ is zero).

\subsection{Action (eikonal) in canonical charts}

Since $\La$ is Lagrangian, it follows that the Pfaff equation
\begin{equation}\label{eq:02-eiko}
    d\tau(\a)=P(\a)\,dX(\a)\equiv
    P_1(\a)\,dX_1(\a)+\dotsm+P_n(\a)\,dX_n(\a)
\end{equation}
is locally solvable on $\La$, and the solution is unique up to an
additive constant. A solution of Eq.~\eqref{eq:02-eiko} is called
an \textit{eikonal} (or \textit{action}). Let $(U,I)$ be a
canonical chart. We define the eikonal in this chart by the formula
\begin{equation}\label{eiko}
    \tau_{(U,I)}(\a)=\int_\ga P(\a)\,dX(\a)
    +\int_{\a_*}^\a P(\a)\,dX(\a),
\end{equation}
where $\ga$ and $\a_*$ are the same as in Sec.~\ref{s14} and the
second integral is taken over an arbitrary path entirely lying in
$U$. The eikonal $\tau_{(U,I)}(\a)$ depends on the choice of the
path $\ga$ (unless the cohomology class of the form $P\,dX$ in
$H^1(\La)$ is trivial).

\subsection{Local canonical operator}

Now we are in a position to write out a formula specifying the
local canonical operator $K_{(U,I)}$ in the canonical chart
$(U,I)$. Let $ a\in C_0^\infty(U)$. Then
\begin{equation}\label{KO-1}
    [K_{(U,I)} a](x,h)=
    \ov\cF_{p_\I\to x_\I}^{1/h}
    \left[\frac{e^{\frac ih[\tau_{(U,I)}(\a)-p_\I X_\I(\a)]
    -i\pi m_{(U,I)}/2}
     a(\a)}{\sqrt{\abs{\cJ_I(\a)}}}\right]_{\a=\a(x_I,p_\I)},
\end{equation}
where $\bar\cF_{p_\I\to x_\I}^{1/h}$ is the inverse $1/h$-Fourier
transform with respect to the variables $p_\I$ (see Sec.~\ref{s6}).
In a nonsingular chart ($\abs{\I}=\varnothing$) the Fourier
transform disappears, and the formula acquires the simpler form
\begin{equation}\label{KO-2}
    [K_{(U,I)} a](x,h)=
    \frac{e^{\frac ih\tau_{(U,I)}(\a)
    -i\pi m_{(U,I)}/2}
     a(\a)}{\sqrt{\abs{\cJ(\a)}}}\bigg|_{\a=\a(x)},
\end{equation}

\subsection{Quantization condition and global canonical operator}

Now assume that the \textit{Bohr--Sommerfeld quantization
conditions}\footnote{These conditions should be understood as
follows: if $\La$ depends on parameters, then, for each $h>0$,
conditions~\eqref{quant-cond} single out the set of admissible
values of these parameters.}
\begin{equation}\label{quant-cond}
 \frac2{\pi h}\oint_{\ga}P(\a)\,dX(\a)\equiv\ind\ga \pmod4
\end{equation}
hold for all cycles $\ga$ on $\La$. (It suffices to require that
\eqref{quant-cond} holds for a basis of independent cycles on
$\La$.)
\begin{theorem}\label{th030}
If the quantization conditions~\eqref{quant-cond} are satisfied,
then the following assertions hold\rom:
\begin{itemize}
    \item The local canonical operators \eqref{KO-1} are
    independent of the choice of the paths~$\ga$
    in Sec.~\rom{\ref{s14}} and coincide modulo $O(h)$ on the
    intersections of the canonical charts.
    \item Let $\{e_{(U,I)}\}$ be a locally finite partition of
unity on $\La$ subordinate to the cover of $\La$ by the canonical
charts $(U,I)$. Define an operator $K_{(\La,d\mu)}^{1/h}$ on
$C_0^\infty(\La)$ by the formula
\begin{equation}\label{global-operator}
   [K_{(\La,d\mu)}^{1/h} a](x,h)=
   \sum_{(U,I)}[K_{(U,I)}(e_{(U,I)} a)](x,h).
\end{equation}
This operator is modulo $O(h)$ independent of the choice of the
canonical atlas and the partition of unity.
\end{itemize}
\end{theorem}
The operator $K_{(\La,d\mu)}^{1/h}$ defined in
\eqref{global-operator} is called \textit{Maslov's canonical
operator} on the Lagrangian manifold $\La$ with measure $d\mu$ and
central point~$\a_0$.

\subsection{Commutation formula and asymptotic solutions}

Consider a differential or pseudodifferential operator $\wh
L=L(\2x,\skew4\1{\wh p},h)\equiv L\Bl(\2x,-\1{ih\pd{}x},h\Br)$ with
smooth symbol $L(x,p,h)=H(x,p)+hL_1(x,p)+\dotsm$. The leading term
$H(x,p)$ of the expansion of $L$ in powers of $h$ is called the
\textit{Hamiltonian}.
\begin{theorem}\label{th51}
Let $ a\in C_0^\infty(\La)$. Then
\begin{equation}\label{e525}
    \wh L K_{(\La,d\mu)}^{1/h}  a
    =K_{(\La,d\mu)}^{1/h}(H|_{\La} a)+O(h),
\end{equation}
where $H|_{\La}=H(X(\a),P(\a))$ is the restriction of $H(x,p)$ to
$\La$. If, moreover, $H|_{\La}\equiv0$ and the measure $d\mu$ is
invariant under shifts along the trajectories of the Hamiltonian
vector field $\xi_H=H_p(x,p)\pa_x-H_x(x,p)\pa_p$, then
\begin{equation}\label{e527}
    \wh LK_{(\La,d\mu)}^{1/h}  a
    =-ihK_{(\La,d\mu)}^{1/h}\bl(\xi_H a -
    \tfrac12(\tr H_{xp})|_{\La} a+iL_1|_{\La} a\br)+O(h^2),
\end{equation}
where $\tr H_{xp}$ is the trace of the matrix $H_{xp}(x,p)$.
\end{theorem}
This theorem suggests a natural way for constructing asymptotic
solutions of the equation $\wh Lu=0$: find a Lagrangian manifold
$\La$ with $H|_\La=0$, equip it with an invariant measure $d\mu$,
and solve the \textit{transport equation} $\xi_H a - \tfrac12(\tr
H_{xp})|_{\La} a+iL_1|_{\La} a=0$; the desired solutions have the
form $u=K_{(\La,d\mu)}^{1/h} a$.

\section{New formulas}\label{s2}

Now let us present new formulas for Maslov's canonical operator.
These formulas differ most dramatically from the standard formulas
in the singular charts (although the expression for the nonsingular
charts acquires a slightly different form as well), and they can be
written out provided that the Lagrangian manifold (or at least the
part of it where we intend to use the new formulas) satisfies
Condition~\ref{co:02-1} below. We point out that our formulas give
functions with the same asymptotics as the standard canonical
operator. (See Theorem~\ref{compare}.) Therefore, the counterparts
of Theorems~\ref{th030} and~\ref{th51} hold for the new expression
of the canonical operator automatically, and that is why we do not
even bother to state or mention them in what follows.

Let $\La$ be a Lagrangian manifold in $\RR_{(x,p)}^{2n}$ equipped
with a measure $d\mu$ and an initial point $\a_0$. Throughout this
section, we assume that the quantization condition
\eqref{quant-cond} is satisfied.

\subsection{Main condition and eikonal coordinates}\label{s21}

From now on, we assume that $\La$ satisfies the following
\begin{condition}\label{co:02-1}
The form $P(\a)\,dX(\a)$ is nonzero for each $\a\in\La$.
\end{condition}
Thus, if $\tau$ is an eikonal in a neighborhood $U$ of some point
of $\La$, then $d\tau\ne0$, and hence (provided that $U$ is
sufficiently small) we can supplement $\tau$ with some functions
$\psi_1,\dots,\psi_{n-1}$ such that
$(\tau,\psi)\equiv(\tau,\psi_1,\dots,\psi_{n-1})$ is a coordinate
system in $U$. A coordinate system of this kind will be called an
\textit{eikonal coordinate system}. The expressions of the
functions $(X(\a),P(\a))$ via eikonal coordinates will be denoted
by\footnote{Rather than the technically correct
$(X(\a(\tau,\psi)),P(\a(\tau,\psi)))$.}
$(X(\tau,\psi),P(\tau,\psi))$ or even simply by $(X,P)$ with the
arguments omitted. The same notation will be used for other
functions on $\La$. The symbol $\mu=\mu(\tau,\psi)$ will from now
on be used to denote the density of the measure $d\mu$ in eikonal
coordinates $(\tau,\psi)$, so that
\begin{equation*}
  d\mu=\mu(\tau,\psi)\,d\tau\wedge
       d\psi_1\wedge \dotsm\wedge d\psi_{n-1}.
\end{equation*}

One can readily prove the following assertion.
\begin{proposition}\label{pr:04-01}
In eikonal coordinates, one has the relations
\begin{gather}\label{eq:04-01}
    \langle P,X_\tau\rangle=1,\qquad
    \langle P,X_{\psi_j}\rangle=0,\qquad
               j=1,\dotsc,n-1,
\\ \label{eq:04-02}
    \langle P_{\psi_j},X_\tau\rangle
       =\langle P_\tau, X_{\psi_j}\rangle,\quad
    \langle P_{\psi_j},X_{\psi_k}\rangle
       =\langle P_{\psi_k},X_{\psi_j}\rangle,\quad
               j,k=1,\dotsc,n-1.
\end{gather}
\end{proposition}

\subsection{Canonical operator in a nonsingular
chart}\label{sec:04-04}

In a nonsingular chart $(U,I)$, $\I=\varnothing$, one still uses
formula~\eqref{KO-2}. The only refinement is that now we have
eikonal coordinates $(\tau,\psi)$, where $\tau=\tau_{(U,I)}$,
instead of the general coordinates $\a$ in $U$, and so we can
partly compute the Jacobian $\cJ$ in the eikonal coordinates using
Proposition~\ref{pr:04-01}. Namely, the following assertion holds.
\begin{proposition}\label{pr:02-02}
In the eikonal coordinates, one has
\begin{equation}\label{eq:04-03}
  \abs{\cJ(\tau,\psi)}=\frac{\sqrt{\smash[b]
  {\det(X_\psi^*(\tau,\psi)X_\psi(\tau,\psi))}}}
  {\abs{\mu(\tau,\psi)}\abs{P(\tau,\psi)}},
\end{equation}
where $X_\psi^*X_\psi=\bl(\langle X_{\psi_j},X_{\psi_k}\rangle\br)$
is the Gram matrix of the vectors $X_{\psi_1},\dotsc,
X_{\psi_{n-1}}$.
\end{proposition}
Accordingly, the expression~\ref{KO-2} for the local canonical
operator becomes
\begin{equation}\label{eq:04-06}
    [K_{(U,I)} a](x,h)=
    e^{\frac ih\tau(x)-i\pi m_{(U,I)}/2} a(\tau,\psi)
    \frac{\sqrt{\abs{\mu(\tau,\psi)}\abs{P(\tau,\psi)}}}
    {\sqrt[4]{\smash[b]{\det(X_\psi^*(\tau,\psi)X_\psi(\tau,\psi))}}}
    \Bigg|\mathstrut_{\substack{\tau=\tau(x)\\\psi=\psi(x)}},
\end{equation}
where $\tau=\tau(x)$, $\psi=\psi(x)$ is the expression of the
eikonal coordinates $(\tau,\psi)$ via the canonical coordinates $x$
in the chart $(U,I)$.

\subsection{Canonical operator near focal points}\label{sec:04-05}

We have defined the action of the canonical operator on functions
$ a\in C_0^\infty(\La)$ whose support does not meet the set
$\Ga\in\La$ of focal points. Now we should define how the canonical
operator acts on functions supported near focal points.  This is
where our construction differs from that the standard one.

Let $\a_*\in\Ga$; i.e., $\cJ_x(\a_*)=0$. We will construct a ``new
singular chart'' in a neighborhood $U$ of $\a_*$ and define the
``new local canonical operator'' on functions $ a\in
C_0^\infty(U)$. We choose and fix some path $\ga$ on $\La$ with
$\ga(0)=\a_0$ and $\ga(1)=\a_*$ and define the eikonal $\tau(\a)$
in a sufficiently small neighborhood of $\a_*$ by
formula~\eqref{eiko}, where the second integral is taken over an
arbitrary path lying in that neighborhood. Next, we supplement the
eikonal with $n-1$ functions $\psi_1,\dotsc,\psi_n$, thus obtaining
a system $(\tau,\psi)$ of eikonal coordinates on $\La$ in a
neighborhood of $\a_*$. The coordinates of $\a_*$ will be denoted
by $(\tau_*,\psi_*)$. Let $k=\rank X_\psi(\tau_*,\psi_*)$. We have
$k<n-1$, because otherwise $\a_*$ would not be a focal point. Take
$k$ linearly independent columns of the matrix
$X_\psi(\tau_*,\psi_*)$ and accordingly divide the variables $\psi$
into two parts $\psi'$ and $\psi''$, the first part including the
variables corresponding to the chosen linearly independent columns,
and the second part including all the other variables.\footnote{If
$k=0$, then $\psi'$ is empty and the formulas given below undergo
obvious modifications. This is always the case for $n=2$.} We
assume (renumbering the variables $\psi$ if necessary) that
$\psi'=(\psi_1,\dotsc,\psi_k)$ and $\psi''=(\psi_{k+1},\dotsc,
\psi_{n-1})$. Note that the Gram matrix $X_{\psi'}^*X_{\psi'}$ is
invertible in a neighborhood of the point $(\tau_*,\psi_*)$.

Consider the system of $k+1$ equations
\begin{equation}\label{me1}
    \langle P(\tau,\psi),x-X(\tau,\psi)\rangle=0, \qquad
    \langle X_{\psi_j}(\tau,\psi),x-X(\tau,\psi)\rangle =0,
    \quad j=1,\dotsc,k.
\end{equation}
\begin{proposition}
System~\eqref{me1} defines smooth functions
\begin{equation}\label{e04}
    \tau=\tau(x,\psi''),\qquad \psi'=\psi'(x,\psi'')
\end{equation}
in a neighborhood of the point $(x_*,\psi''_*)$, where
$x_*=X(\tau_*,\psi_*)$, such that $\tau_*=\tau(x_*,\psi_*'')$ and
$\psi_*'=\psi'(x_*, \psi_*'')$. Moreover, there exists a
neighborhood $W$ of the point $(x_*,\psi''_*)\in\RR^{2n-1-k}$ such
that the following conditions hold\rom:

\rom{(i)} The differentials
$d\tau_{\psi_{k+1}},\dotsc,d\tau_{\psi_{n-1}}$ are linearly
independent at each point of the set
\begin{equation*}
  \Pi=\{(x,\psi'')\in W\colon \tau_{\psi''}(x,\psi'')=0\},
\end{equation*}
which is therefore an $n$-dimensional submanifold.

\rom{(ii)} The image $U$ of $\Pi$ under the mapping
$(x,\psi'')\longmapsto (x,\tau_x(x,\psi''))$ is contained in $\La$
and is a neighborhood of the point $\a_*$ in $\La$.

\rom{(iii)} For $(x,\psi'')\in W$, one has $\det
M(\tau(x,\psi''),\psi'(x,\psi''),\psi'')\ne0$, where
\begin{equation}\label{M-matrix}
    M=\begin{pmatrix}
      P & X_{\psi'} & P_{\psi''}
      -P_{\psi'}(X_{\psi'}^*X_{\psi'})^{-1}X_{\psi'}^*X_{\psi''}
    \end{pmatrix}.
\end{equation}
\end{proposition}
The domain $U\subset \La$, together with the eikonal coordinates
$(\tau,\psi)$ and the functions \eqref{e04}, will be called a
\textit{new singular chart} on $\La$. (Without loss of generality,
we can assume that both $U$ and $W$ are connected and simply
connected.)

We define the \textit{index} $m_U$ of the new singular chart by
setting
\begin{equation}\label{MasIndNew}
    m_U=\frac1\pi\BL(\arg\cJ^\e(\a_0)\big|_{\e=0}^{\e=1}
    +\var_\ga\arg\cJ^1(\a)
    -\sum_{s=1}^{2n-k-1}\arg\la_s\BR),
\end{equation}
where $\ga$ is the same path as above, the $\la_s$ are the
eigenvalues of the $(2n-k-1)\times(2n-k-1)$ matrix $\begin{pmatrix}
      E-i\tau_{xx}(x_*,\psi''_*) &
      -i\tau_{x\psi''}(x_*,\psi''_*) \\
      -i\tau_{\psi'' x}(x_*,\psi''_*) &
      -i\tau_{\psi''\psi''}(x_*,\psi''_*)
    \end{pmatrix}$,
and $\arg\la_j\in\bl[-\frac\pi2,\frac\pi2\br]$.

For $ a\in C_0^\infty(U)$, set
\begin{multline}\label{e05}
    [K_U^{1/h} a](x,h)
    =\frac{e^{-i\pi
    m_U/2}}{(2\pi h)^{\frac{n-1-k}2}}
    \int
    \frac{e^{i\frac\tau h} a\sqrt{\left\vert\mu\det M\right\vert}}
    {\sqrt{\det(X_{\psi'}^*X_{\psi'})}}
    \bigg|_{\substack{\tau=\tau(x,\psi'')\\
    \psi'=\psi'(x,\psi'')}}\chi(x,\psi'')\,d\psi'',
\end{multline}
where $M$ is the matrix \eqref{M-matrix} and $\chi(x,\psi'')$ is a
smooth cutoff function on $W$ such that $\chi=1$ on $\Pi$ and
$\chi(x,\cdot)$ is compactly supported in
$W_x=\{\psi''\colon(x,\psi'')\in W\}$ for each $x$. For $k=0$ (the
coordinates $\psi'$ are absent), formula~\eqref{e05} becomes
\begin{equation}\label{e06}
    [K_U^{1/h} a](x,h)
    =\frac{e^{-i\pi
    m_U/2}}{(2\pi h)^{\frac{n-1-k}2}}
    \int \left[e^{i\frac\tau h}  a
    \sqrt{\abs{\mu\det(P,P_{\psi})}}
    \right]_{\tau=\tau(x,\psi)}\chi(x,\psi)\,d\psi.
\end{equation}

\subsection{Comparison with the standard canonical operator}

Let us compare the canonical operator $K_U^{1/h}$ constructed in
Sec.~\ref{sec:04-05} with the standard local canonical operator.
Without loss of generality, we assume that the domain $U$ is
sufficiently small and hence can be covered with a single ``old''
canonical chart with coordinates $(x_I,p_\I)$. Consider the
canonical operator $K_{(U,I)}$ defined by formula~\eqref{KO-1},
where the eikonal $\tau_{(U,I)}$ and the index $m_{(U,I)}$ are
defined with the use of the same path $\ga$ as in the construction
of $K_U^{1/h}$. Then the following assertion holds.
\begin{theorem}\label{compare}
Under these assumptions, one has
\begin{equation}\label{ravno}
    K_U^{1/h} a=K_{(U,I)} a+O(h)
    \qquad\text{for every $ a\in C_0^\infty(U)$.}
\end{equation}
Moreover, for each $ a\in C_0^\infty(U)$ there exist $ a_j\in
C_0^\infty(U)$, $j=1,2,\dotsc$, such that
\begin{equation}\label{ravno-N}
    K_U^{1/h} a=K_{(U,I)}\BL( a
    +\sum_{j=1}^{N-1}h^j a_j\BR)+O(h^N),
    N=1,2,\dotsc\,.
\end{equation}
\end{theorem}
\begin{proof}
[Proof \rm is based on Theorem~\ref{atheorem1} in the Appendix] It
is rather technical, and we omit the lengthy computations.
\end{proof}

\subsection{Closing remark}

Condition~\ref{co:02-1} is actually not restrictive, because if it
is violated in a specific problem, then one can introduce an
additional variable~$x_{n+1}$ (a~cyclic variable) on which the
Hamiltonian does not depend and consider solutions of the form
$v(x,x_{n+1})=u(x)e^{\frac ihx_{n+1}}$, where $u(x)$ is the desired
solution of the original problem. If $\La_u$~is the Lagrangian
manifold corresponding to~$u$, then the Lagrangian manifold
corresponding to~$v$ has the form
\begin{equation*}
    \La_v=\La_u\times\{(x_{n+1},p_{n+1})\in\RR^2\colon
    p_{n+1}=1,\;x_{n+1}\text{ is arbitrary}\},
\end{equation*}
and one can readily see that the 1-form
\begin{equation*}
    (p_1\,dx_1+\dotsm+p_{n+1}\,dx_{n+1})|_{\La_v}=
    dx_{n+1}+(p_1\,dx_1+\dotsm+p_{n}\,dx_{n})|_{\La_u}
\end{equation*}
is nonzero on $\La_v$. Thus, this uniformization procedure always
permits one to ensure that Condition~\ref{co:02-1} is satisfied. We
do not further elaborate on the topic.

\section{Examples}\label{s3}

\subsection{Bessel function of order $1/2$}

Consider the Lagrangian manifold
\begin{equation}\label{3D-1}
    \La^3=\{(x,p)\in\RR^6\colon x=X(\tau,\om),\,p=P(\tau,\om),\;
    \tau\in\RR,\,\om\in\SS^2\},
\end{equation}
where $X(\tau,\om)=\tau\bn(\om)$, $P(\tau,\om)=\bn(\om)$, and if
$\om\in\SS^2$ is represented by the spherical coordinates,
$\om=(\theta,\psi)$, then $\bn(\om)\equiv\bn(\theta,\psi)
=(\sin\theta\cos\psi,\sin\theta\sin\psi,\cos\theta)$. We equip
$\La^3$ with the measure $d\mu=d\tau\wedge d\om$, where $d\om$ is
the surface area element of the unit sphere $\SS^2$. In the
spherical coordinates,
\begin{equation}\label{3D-2}
    d\mu=\mu(\tau,\theta,\psi)\,d\tau\wedge d\theta\wedge d\psi,
    \quad
    \mu=\sin\theta.
\end{equation}
Obviously,
\begin{equation}\label{3D-3}
    P(\tau,\om)\,dX(\tau,\om)=d\tau,
\end{equation}
so that $(\tau,\om)$ are eikonal coordinates on $\La$. Next, the
equation
\begin{equation}\label{3D-4}
    \langle P(\tau,\om),x-X(\tau,\om)\rangle=0
\end{equation}
is uniquely solvable for $\tau$,
\begin{equation}\label{3D-5}
    \tau(x,\om)=\langle x,\bn(\om)\rangle,
\end{equation}
and the Jacobian
\begin{equation}\label{3D-6}
    \det(P,P_\om)=\det(P,P_\theta,P_\psi)=
    \det\begin{pmatrix}
      \sin\theta\cos\psi & \cos\theta\cos\psi & -\sin\theta\sin\psi \\
      \sin\theta\sin\psi & \cos\theta\sin\psi & \sin\theta\cos\psi \\
      \cos\theta & -\sin\theta & 0 \\
    \end{pmatrix}=\sin\theta
\end{equation}
is nonzero except for $\theta=0,\pi$. (These singularities are
however artificial: they are due to the degeneration of spherical
coordinates at these points and would not occur if one uses
different spherical coordinates to represent $\om$.) Thus, the
entire manifold $\La^3$ is covered by one singular chart with
coordinates $(\tau,\om)$, and the canonical operator on it has the
form
\begin{equation}\label{3D-7}
\begin{split}
    [K_{\La^3} a](x,h)&=-\frac1{2\pi h}\int
     \int \left[e^{\frac ih\tau}  a \sqrt{\abs{\mu\det(P,P_{\om})}}
    \right]_{\tau=\tau(x,\om)}\chi(x,\om)\,d\theta\wedge d\psi
    \\
    &=-\frac1{2\pi h}\int
     \int e^{\frac ih\langle x,\bn(\om)\rangle}
      a(\langle x,\bn(\om)\rangle,
     \theta,\psi) \sin\theta\,d\theta\wedge d\psi.
\end{split}
\end{equation}
Take $ a\equiv1$; then the integral depends only on $\abs{x}$, and
by taking $x=(0,0,\abs{x})$ we obtain
\begin{multline}\label{3D-8}
    [K_{\La^3}1](x,k)=-\frac1h
    \int_0^\pi e^{\frac ih\abs{x}\cos\psi}\sin\psi\,d\psi
    =\frac{1}{i\abs{x}}e^{\frac ih\abs{x}\cos\psi}
    \big|_{\psi=0}^{\psi=\pi}
    \\=-\frac{2\sin(\abs{x}/h)}{\abs{x}}=-\sqrt{\frac{2\pi
    }{h\abs{x}}}J_{\frac12}(\abs{x}/h),
\end{multline}
where
\begin{equation}\label{3D-9}
    J_{\frac12}(z)=\sqrt{\frac2{\pi z}}\sin z
\end{equation}
is the Bessel function of order $1/2$.

\subsection{Bessel type beams in three-dimensional space}

Consider the Lagrangian manifold
\begin{equation}\label{bess-0}
\Lambda^{3}_0=\{(x,p)\in\RR^6\colon x=X_0(\a,\psi,\phi),\;
             p=P_0(\a,\psi,\phi),\;\;\a,\phi\in\RR,\;\psi\in\SS^1\},
\end{equation}
where
\begin{equation}\label{bess-1}
    X_0(\a,\psi,\phi)=\begin{pmatrix}
      \mathbf{n}(\psi)\alpha \\
      \phi
    \end{pmatrix}, \;
    P_0(\a,\psi,\phi)=\begin{pmatrix}
      \lambda(\phi)\mathbf{n}(\psi) \\
      \alpha\lambda'(\phi)+k
    \end{pmatrix},\quad
    \bn(\psi)=\begin{pmatrix}
      \cos\psi \\
      \sin\psi
    \end{pmatrix},
\end{equation}
$\lambda(\phi)>0$ is a smooth function (for example,
$\lambda=a(1+\tanh \phi)+b$, $a,b>0$) and $k$ is a constant. In what
follows, we use the notation $x_\perp=(x_1,x_2)$ and
$p_\perp=(p_1,p_2)$.

One can readily see that the cycle $\Ga_0\subset\La_0^3$ of
singularities (the set where the projection $\La_0^3\to\RR_x^3$
degenerates) is given by the equation $\alpha=0$; thus, $\Ga_0$ is
a two-dimensional cylinder, and the caustic---the projection
$\gamma_0$ of $\Ga_0$ onto the configuration (physical) space
$\mathbb{R}^3_{x}$---is the $x_3$-axis. It is \textit{degenerate}
in that to each point in $\ga_0$ there corresponds a continuum
(namely, a circle) of focal points. The manifold $\Lambda^{3}_0$ is
diffeomorphic to the product of the circle $\SS^1$ by the plane
$\RR^2$, and the projection
$\Lambda^{3}_0\setminus\Ga_0\to\mathbb{R}^3_{x}\setminus\ga_0$ is a
double covering. We equip $\Lambda^{3}_0$ with the measure
\begin{equation}\label{me2}
    d\mu_0=\la^{-1}(\phi)\,d\a\wedge
    d\phi\wedge d\psi.
\end{equation}

Let us construct the canonical operator on the manifold $\La_0^3$
with the measure $d\mu_0$. First, let us find eikonal coordinates.
One has
\begin{equation*}
    P_0\,dX_0=\lambda(\phi)\,d\alpha +\alpha\lambda'(\phi)\,d\phi+k\,d\phi
         = d(\lambda(\phi)\alpha+k\phi),
\end{equation*}
and we can define the eikonal to be
\begin{equation}\label{eiko2}
    \tau=\lambda(\phi)\alpha+k\phi
\end{equation}
and the eikonal coordinates to be $(\tau,\phi,\psi)$. These
coordinates are in fact global on $\La_0^3$, because the passage to
the original coordinates is always possible by the formula
$\alpha=(\tau-k\phi)/\lambda(\phi)$. We have
\begin{equation}\label{koe-chto}
    \pd{\a}{\tau}=\frac1{\la(\phi)},\qquad\pd{\a}{\phi}
    =-\frac{k+\alpha\lambda'(\phi)}{\lambda(\phi)}.
\end{equation}
The measure in the eikonal coordinates is
\begin{equation}\label{me2a}
    d\mu_0=\la^{-2}(\phi)\,d\tau\wedge
    d\phi\wedge d\psi.
\end{equation}

We claim that $\La_0^3$ is covered by a single singular chart with
the functions $\tau=\tau(x,\psi)$ and $\phi(x,\psi)$ computed as
follows. (Thus, the role of $\psi'$ and $\psi''$ in \eqref{e04} is
played by $\phi$ and $\psi$, respectively, in our example.) Indeed,
let us write out Eqs.~\eqref{me1} for our case. They read
\begin{equation}\label{me1a}
    \langle
    \bn(\psi),x_\perp\rangle=\frac{\tau-k\phi}{\la(\phi)},\qquad
    x_3=\phi,
\end{equation}
and we obtain the global solutions
\begin{equation}\label{me1b}
    \tau(x,\psi)=\la(x_3)\langle\bn(\psi),x_\perp\rangle+kx_3,\qquad
    \phi(x,\psi)=x_3.
\end{equation}
Let us compute the determinant of the matrix \eqref{M-matrix} and
other objects occurring in formula \eqref{e05} as applied to our
case. In the eikonal coordinates, one can readily prove the
orthogonality relations
\begin{equation}\label{orth}
    \langle X_{0\phi},X_{0\psi}\rangle=\langle X_{0\phi},P_{0}\rangle=
    \langle X_{0\phi},P_{0\psi}\rangle=\langle
    P_0,P_{0\psi}\rangle=0,
\end{equation}
and hence, after some computations, we obtain
\begin{align*}
    \det M&=\det\bigl(P_0,X_{0\phi},P_{0\psi}\bigr)=
    \la^2(x_3)+(\la'(x_3)\langle\bn(\psi),x_\perp\rangle+k)^2,\\
    \det(X_{0\phi}^*X_{0\phi})&=\langle X_{0\phi},X_{0\phi}\rangle
    =\frac{1}{\lambda^2(x_3)}\bl[\la^2(x_3)+(\la'(x_3)\langle\bn(\psi),x_\perp\rangle+k)^2\br].
\end{align*}
Thus, the expression \eqref{e05} for the canonical operator becomes
\begin{equation}\label{beam2}
 [K_{\La^3}a](x,h)= \sqrt{\frac{i}{2\pi h}}e^{\frac{i}{h}kx_3}
 \int_0^{2\pi}e^{\lambda(x_3)\langle x_\perp,
 \mathbf{n}(\psi)\rangle}
  a(\la(x_3)\langle\bn(\psi),x_\perp\rangle+kx_3,x_3,\psi)\,d\psi.
\end{equation}
If $a(\tau,\phi,\psi)$ actually has the form
$a=a(\a,\phi)$ and is even in $\a$, then we obtain
\begin{equation}\label{f}
    [K_{\La^3}a](x,h)= \sqrt{\frac{2\pi i}{h}}a(\abs{x_\perp},x_3)e^{\frac{i}{h}kx_3}
    \mathcal{J}_0\BL(\frac{\lambda(x_3)\abs{x_\perp}}{h}\BR)+O(h),
\end{equation}
where $J_0$ is the Bessel function of zero order. If $a=\const$ and
$\la=\const$, then this function, which acquires the form
$\const\cdot e^{\frac{i}{h}kx_3}
    \mathcal{J}_0\bl(\frac{\la\abs{x_\perp}}{h}\bl)$,
is know as the \textit{Bessel beam} in optics (see \cite{Kisel}, where further references can be found). If $a$ is a compactly supported function, then the function \eqref{beam2} is a pulse oscillating in the $x_3$-direction and having the shape of a Bessel function in the variables $(x_1,x_2)$ with scale $\la(x_3)$ depending on $x_3$.
Now consider the problem of the beam evolution according to the
three-dimensional wave equation
\begin{equation}\label{WE}
\frac{1}{c^2}\frac{\pa^2 u}{\pa t^2}=\triangle
u\equiv\biggl(\frac{\pa^2 u}{\pa x_1^2}+\frac{\pa^2
 u}{\pa x_2^2}+\frac{\pa^2 u}{\pa x_3^2}\biggr).
\end{equation}
More precisely, we wish to study the solution of the corresponding
Cauchy problem for this equation. (For $\la,\a=\const$, the function $e^{\frac{i}{h}(\sqrt{\la^2+k^2}+kx_3)}
    \mathcal{J}_0\bl(\frac{\la\abs{x_\perp}}{h}\bl)$ is the exact solution of this equation.) It is well known that the
solution of the wave equation splits into two parts describing
waves running in opposite directions. We restrict ourselves to a
beam propagating in one direction. To this end, we factorize the
wave operator,
\begin{equation*}
  -h^2\pd{^2}{t^2}+h^2\triangle=\biggl(-ih\frac{\pa}{\pa t}-c \sqrt{-h^2\triangle} \biggr)
  \biggl(-ih\frac{\pa}{\pa t}+c \sqrt{-h^2\triangle} \biggr),
\end{equation*}
and consider the Cauchy problem for an equation of first order in
time,
\begin{equation}\label{beam3}
 -ih\frac{\pa u}{\pa t}+c \sqrt{-h^2\triangle} u=0,
u|_{t=0}=u_0,
\end{equation}
where $u_0$ has the form \eqref{beam2}. This problem can be solved
by means of the Fourier transform, which gives the answer in a form
of rapidly oscillating integral. (One can also pass to the cylindrical coordinates in the equation, but these coordinates produce a fictitious singularity on the $x_3$-axis, which is inconvenient.) The study of this integral is not
trivial, and in our opinion it is much easier to use the Maslov
canonical operator to describe this solution.  By applying
Maslov's general theory, we see that the asymptotic solution
\begin{equation}\label{Schr4}
    u=K_{\La_t}[a(t)],
\end{equation}
of \eqref{beam3} is based on the family of the Lagrangian manifolds
$\Lambda_t=g^t_H\Lambda_0^3$ obtained by the shift of the manifold
$\La_0^3$ with the help of the phase flow $g_H^t$ corresponding to
the Hamiltonian $H=c|p|\equiv c\sqrt{p_1^2+p_2^2+p_3^2}$. Since
this Hamiltonian does not depend on $x$, we can readily integrate
the corresponding Hamilton system and obtain
 \begin{gather} \label{bess-v} \Lambda^{3}_t=\{p_\perp=\mathbf{n}(\psi)\mathcal{P},\,p_3=P_3,\,
x_\perp=\mathbf{n}(\psi)\mathcal{X},\, x_3=X_3,\}\\
\mathcal{P}=\lambda(\phi),P_3=\alpha \lambda'(\phi)+k,
|P|=\sqrt{\lambda^2+(\alpha \lambda'(\phi)+k)^2},
\nonumber\\\mathcal{X}=\alpha+t\frac{\mathcal{P}}{|P|}=
\alpha+tc\frac{\lambda(\phi)}{\sqrt{\lambda^2+(\alpha\lambda'(\phi)+k)^2}},
 X_3=\phi+tc\frac{P_3}{|P|}, \nonumber
\end{gather}
Since the Hamiltonian is of first order in the momenta, it follows
that the eikonal is constant along the trajectories and is still
given by formula \eqref{eiko2} if we use the coordinates
$(\a,\phi,\psi)$ brought to $\La_t$ from $\La_0^3$ by the
Hamiltonian flow. In particular, we still have the global eikonal
coordinates $(\tau,\phi,\psi)$. In these coordinates,
\begin{equation*}
 X_\psi=\begin{pmatrix}\mathbf{n}_\perp\mathcal{X} \\0 \end{pmatrix},
 \qquad
 X_\phi=\begin{pmatrix}\mathbf{n}(\mathcal{X}_\phi+
   \mathcal{X}_\alpha\alpha_\phi)
 \\X_{3\phi}+X_{3\alpha}\alpha_\phi) \end{pmatrix}
\end{equation*}
(where the derivatives on the right-hand side are taken in the
coordinate system $(\a,\phi,\psi)$); moreover, the orthogonality
relations similar to \eqref{orth} hold, and so we have
\begin{equation*}
 \begin{pmatrix}X_\psi& X_\phi
  \end{pmatrix}^*\begin{pmatrix}X_\psi& X_\phi
  \end{pmatrix}
  =\begin{pmatrix}\mathcal{X}^2& 0\\0&X_\phi^2
  \end{pmatrix}.
\end{equation*}
and
\begin{equation}\label{Jac1}
  \det \Bigl(\begin{pmatrix}X_\psi& X_\phi
  \end{pmatrix}^*\begin{pmatrix}X_\psi& X_\phi
  \end{pmatrix}\Bigr)=\mathcal{X}^2\bigl((\mathcal{X}_\phi+
   \mathcal{X}_\alpha\alpha_\phi)^2+(X_{3\phi}+X_{3\alpha}\alpha_\phi)^2
  \bigr).
\end{equation}
Hence according to \eqref{eq:04-03} the focal points are defined by
the equations
\begin{align}\label{Jac2}
  \mathcal{X}&=0
\\\label{Jac3}(\mathcal{X}_\phi+
\mathcal{X}_\alpha\alpha_\phi)&=0,\quad
X_{3\phi}+X_{3\alpha}\alpha_\phi=0.
\end{align}
After some transformations, we find that the two equations in
\eqref{Jac3} are equivalent to the single equation
\begin{equation}\label{Jac4}
\bigl(\lambda^2(\phi)+(\alpha \lambda'(\phi)+k)^2\bigr)^2-
2ct\biggl(\lambda'(\phi)k+\alpha\biggl({\lambda'(\phi)}^2-\frac{\lambda(\phi)\lambda''(\phi)}{2}\bigg)\biggr)=0.
\end{equation}
The projection of the singularities defined by \eqref{Jac2} is the
$x_3$-axis  $x_1=x_2=0$. We assume for simplicity that Eqs.
\eqref{Jac4} do not hold on the support of the function
$a(\tau,\phi,\psi)$ for $t\in[0,T]$. By virtue of the orthogonality
relations \eqref{orth} and the formula
$|P_\psi|=\mathcal{P}=\lambda$, we have
\begin{equation}\label{Jac4aa}
    |\det M|=|P||P_\psi||X_\phi|=\biggl||P|^2-\frac{ct\mathcal{P}}{|P|}\biggl(2\frac{P_3}{|P|}
\lambda'- \alpha\frac{\mathcal{P}}{|P|}\lambda''\biggr)\biggr|.
\end{equation}
 Equations~\eqref{me1} read
\begin{equation}\label{tauphi1}
    \mathcal{P}(\langle\mathbf{n}(\psi),x_\perp\rangle-\mathcal{X})+
    P_3(x_3-X_3)= \mathcal{X}_\phi(\langle\mathbf{n}(\psi),x_\perp\rangle-\mathcal{X})+
    X_{3\phi}(x_3-X_3)=0.
\end{equation}
Since the vectors $P$ and $X_\phi$ are orthogonal, it follows that
so are the vectors $\begin{pmatrix} \mathcal{P}\\P_3
\end{pmatrix}$ and $\begin{pmatrix} \mathcal{X}_\phi\\X_{3\phi}
\end{pmatrix}$. Hence the determinant $\mathcal{P}X_{3\phi}-P_3\mathcal{X}_\phi$
of system \eqref{tauphi1} is up to the sign equal to the product of
norms of the latter vectors. This means that the determinant of system
\eqref{tauphi1} is zero if and only if \eqref{Jac3} is true, which
contradicts our assumption. Thus, Eqs.~\eqref{tauphi1} are
equivalent to the equations
\begin{equation*}
   \mathcal{X}= q,\quad X_3=x_3,\qquad\text{where}\quad
   q=\langle\mathbf{n}(\psi),\,x_\perp\rangle.
\end{equation*}
Let us treat these equations as a system for the unknown variables
$(\alpha, \phi)$,
\begin{equation}\label{tauphi2}
\alpha+ct\frac{\mathcal{P}}{|P|}=q,\quad \phi+ct\frac{P_3}{|P|}=x_3.
\end{equation}
For small $t$, this system is obviously solvable by the implicit
function theorem. We denote the solution by $\phi=\ph(q,x_3,t)$,
$\a=\a(q,x_3,t)$, $\tau=\la\a+k\phi\equiv\tau(q,x_3,t)$. Of course,
we cannot write out the solution explicitly, but we still can
obtain useful relations for it. For example, by eliminating
$1/|P|$, we obtain
\begin{equation*}
 (q-\alpha)^2+(x_3-\phi)^2=t^2c^2.
\end{equation*}
From this equation and \eqref{tauphi2}, we find that
\begin{gather}\label{alpha1}
   \alpha=q-\sqrt{R},\qquad |P|=\frac{ct\lambda}{\sqrt{R}},
   \qquad R(\phi,x_3,t)=t^2c^2-(x_3-\phi)^2,\\ \label{tau2}
   \tau=\lambda(\phi)(q-\sqrt{R(\phi,x_3,t)})+k\phi,
\end{gather}
which can be helpful. Now we can write the solution using the
canonical operator:
\begin{equation}\label{example}
    u(x,t,h)=\BL(\frac{i}{2\pi h}\BR)^{1/2}
    \int \left[e^{i\tau/h}
   a(\tau,\phi,\psi)\sqrt{{\frac{\la\abs{P}}{\abs{X_\phi}}}}
    \right]_{\substack{\phi=\ph(q,x_3,t)\\\tau=\tau(q,x_3,t)\\
    q=\mathbf{\langle n(\psi),x_\perp\rangle}}}\,d\psi.
\end{equation}
Note that it can be proved that for small $x_\perp$ this solution
still has the asymptotics equal to the product of an oscillating
exponential in the $x_3$-direction and the Bessel function in the
$x_\perp$-direction. The proof of this fact, as well as further simplification of the integral~\eqref{example}, requires much place, and here we do not dwell on the topic.

\section{Appendix. Fourier integrals and the canonical
operator}\label{s4}

In this short text, we discuss oscillatory integrals of the form
\eqref{FIO} below with parameter $h\to0$ and show how such
integrals (which are the counterpart of the Fourier integral
distributions introduced by H\"ormander \cite{Hor6}) are related to
Maslov's canonical operator (see \cite{Mas1} and also \cite{MSS1}).
The relationship between Fourier integral operators and Maslov's
canonical operator was also discussed at length in~\cite{NOSS,Yo}.

\subsection{Nondegenerate phase functions}

\begin{definition}
A real-valued function $\Phi(x,\theta)$ defined on an open set
$V\subset\RR_x^{n}\times\RR_{\theta}^m$ is called a
\textit{nondegenerate phase function} if the differentials
$d(\Phi_{\theta_1}),\dotsc,d(\Phi_{\theta_m})$ are linearly
independent on the set
\begin{equation}\label{e506}
    C_\Phi=\bl\{(x,\theta)\in V\colon \Phi_\theta(x,\theta)=0\br\},
\end{equation}
or, equivalently,
\begin{equation}\label{e505}
    \rank\begin{pmatrix}
      \Phi_{\theta x}(x,\theta) & \Phi_{\theta\theta}(x,\theta) \\
    \end{pmatrix}
        =m,\qquad (x,\theta)\in C_\Phi.
\end{equation}
\end{definition}
By the implicit function theorem, $C_\Phi$ is a smooth
$n$-dimensional manifold, and the functions~$\Phi_\theta$ can serve
as coordinates in the directions transversal to~$C_\Phi$.
\begin{proposition}
The mapping
\begin{equation}\label{e507}
         j_\Phi\colon C_\Phi\lra\RR_{(x,p)}^{2n},\quad
         (x,\theta)\longmapsto
              \bl(x,\Phi_x(x,\theta)\br),
\end{equation}
is a local diffeomorphism of $C_\Phi$ onto its image
$\La_\Phi=j_\Phi(C_\Phi)\subset\RR_{(x,p)}^{2n}$, which is a
Lagrangian submanifold \rom(possibly with self-intersections\rom).
\end{proposition}
\begin{proof}
Let $v={}^t(\eta,\xi)$ be a vector tangent to $C_\Phi$; thus,
$\Phi_{\theta x}\eta+\Phi_{\theta\theta}\xi=0$. Let
$j_{\Phi*}(v)\equiv(\eta,\Phi_{x\theta}\xi)=0$. Then
$\Phi_{\theta\theta}\xi=-\Phi_{\theta x}\eta=0$, and
${}^t\xi\begin{pmatrix} \Phi_{\theta x}(x,\theta) &
\Phi_{\theta\theta}(x,\theta) \\ \end{pmatrix}=0$, whence $\xi=0$
by \eqref{e505} and so $v=0$. This proves that $j_\Phi$ is a local
diffeomorphism onto the image of itself. Next,
\begin{multline*}
    dx\wedge dp=dx\wedge(\Phi_{x\theta}d\theta+\Phi_{xx}dx)
    =(\Phi_{\theta x}dx)\wedge d\theta
    =(\Phi_{\theta x}dx+\Phi_{\theta \theta}d\theta)\wedge d\theta
    \\=d(\Phi_\theta)\wedge d\theta=0
\end{multline*}
on the tangent space to $\La_\Phi$, because $\Phi_\theta=0$ on
$C_\Phi$ and the products $dx\wedge(\Phi_{xx}dx)=0$ and
$\Phi_{\theta \theta}d\theta \wedge d\theta$ are zero by the
antisymmetry of the exterior product.
\end{proof}
\begin{definition}\label{def-loc-det}
A nondegenerate phase function $\Phi(x,\theta)$ is called a
(\textit{local}) \textit{determining function} of a Lagrangian
manifold $\La\subset\RR^{2n}$ if $\La_\Phi\subset\La$.
\end{definition}
Under the conditions of Definition~\ref{def-loc-det}, $\La_\Phi$
does not have self-intersections and is (relatively) open in $\La$.
In what follows, we identify $C_\Phi$ with $\La_\Phi$ via the
mapping
\begin{equation}\label{e507a}
  j_\Phi\colon C_\Phi\lra\La_\Phi,
\end{equation}
so that functions (differential forms) on $C_\Phi$ are
automatically functions (differential forms) on $\La_\Phi$ and vice
versa. The following proposition provides a good example of this.
\begin{proposition}\label{apropo2}
The function $\tau=\Phi|_{C_\Phi}$ is an action on $\La_\Phi$,
i.e., satisfies
\begin{equation}\label{a-action}
    d\tau=p\,dx \big|_{\La_\Phi}.
\end{equation}
\end{proposition}
\begin{proof}
One has
\begin{equation*}
    d\tau=d(\Phi|_{C_\Phi})=(d\Phi)|_{C_\Phi}=(\Phi_x\,dx+\Phi_\theta\,d\theta)|_{C_\Phi}
    =p\,dx|_{\La_\Phi},
\end{equation*}
because $\Phi_x=p$ and $\Phi_\theta=0$ on $C_\Phi$.
\end{proof}

\subsection{Fourier integrals}

Let $\Phi(x,\theta)$, $(x,\theta)\in
V\subset\RR_x^{n}\times\RR_{\theta}^m$,  be a nondegenerate phase
function, and let $a\in C_0^\infty(V)$ be a smooth compactly
supported function.
\begin{definition}\label{def-FI}
The function
\begin{equation}\label{FIO}
    \cI[\Phi,a](x,h)=\frac{e^{i\pi m/4}}{(2\pi h)^{m/2}}
    \int e^{\frac ih\Phi(x,\theta)}a(x,\theta)\,d\theta
\end{equation}
is called the \textit{Fourier integral with phase function $\Phi$
and amplitude $a$}.
\end{definition}
Note that $C_\Phi$ is the set of stationary points of the integral
\eqref{FIO}, and hence the asymptotics of $\cI[\Phi,a]$ as $h\to0$
depends only on the behavior of $a$ near $C_\Phi$. In particular,
we can assume that $a$ is supported in an arbitrarily small
neighborhood of~$C_\Phi$. Now if $a|_{C_\Phi}=0$, then $a$ can be
represented as a linear combination
$a=\sum_{j=1}^ma_j\Phi_{\theta_j}$ of the derivatives
$\Phi_{\theta_j}$ with coefficients $a_j\in C_0^\infty(V)$.
Integration by parts gives $\cI[\Phi,a]=ih\sum_{j=1}^m
\cI[\Phi,a_{j\theta_j}]$, and we see that the Fourier
integral~\eqref{FIO} modulo $O(h)$ depends only on the restriction
of the amplitude to $C_\Phi$. \textit{In what follows, we are only
interested in the leading term of the asymptotics, and accordingly
we only describe the restriction of the amplitude to $C_\Phi$\rom;
its continuation outside $C_\Phi$ can be chosen
arbitrarily.}\footnote{The subsequent terms of the asymptotics can
be studied as well, but we do not dwell on the topic.}

\subsection{Canonical operator as a Fourier integral}

Let $\La$ be a Lagrangian manifold in $\RR_{(x,p)}^{2n}$, and let
$(U,I)$, $I\subset\{1,\dotsc,n\}$, be a canonical chart on~$\La$.
In particular, $U$ is given by Eqs.~\eqref{loc-coor0}. The local
canonical operator in $(U,I)$ is given by \eqref{KO-1}, where
$\bar\cF_{p_\I\to x_\I}^{1/h}$ is the inverse $1/h$-Fourier
transform~\eqref{FT1} with respect to the variables $p_\I$. By
substituting~\eqref{FT1} into~\eqref{KO-1}, we see that the local
canonical operator has the form \eqref{FIO} with
\begin{equation}\label{PhP}
    \Phi(x,\theta)=\tau_{(U,I)}(\a(x_I,\theta))
    +\theta(x_\I-X_\I(x_I,\theta)),
\end{equation}
where we have denoted $p_\I$ by $\theta$ to conform in notation
with \eqref{FIO}. We have
\begin{equation*}
    d\Phi= P_I\,dx_I+\theta\,dX_\I+\theta(dx_\I
    -dX_\I)+(x_\I-X_\I)d\theta=
    P_I\,dx_I+\theta\,dx_\I+(x_\I-X_\I)d\theta
\end{equation*}
(the arguments $(x_I,\theta)$ are omitted),
$\Phi_x(x,\theta)=(P_I(x,\theta),\theta)$, and
$\Phi_\theta(x,\theta)=x_\I-X_\I(x_I,\theta)$. In particular,
$\Phi_{\theta x_\I}$ is the $\abs{\I}\times\abs{\I}$ identity
matrix, and so $\Phi$ is a nondegenerate phase function. Next,
$C_\Phi$ is given by the equations $x_\I=X_\I(x_I,\theta)$, and
$\La_\Phi$ is given by the equations $p_I=P_I(x_I,\theta)$,
$p_\I=\theta$, $x_\I=X_\I(x_I,\theta)$, or (eliminating the
variables $\theta$) by Eqs.~\eqref{loc-coor0}. We conclude that
$\La_\Phi=U\subset\La$, and hence the standard
representation~\eqref{KO-1} of the canonical operator in the chart
$(U,I)$ is none other than a special case of the general Fourier
integrals~\eqref{FIO}.

\subsection{Jacobians in Fourier integrals}

The amplitude of the Fourier integral \eqref{KO-1} representing the
local canonical operator has the form of the product
\begin{equation*}
    a(x,p_\I)=\frac{\ph(\a(x_I,p_\I))}{\sqrt{\cJ_I(\a(x_I,p_\I))}}.
\end{equation*}
The factor $1/\sqrt{\cJ_I(\a(x_I,p_\I))}$ plays an important role
when comparing local canonical operators in different canonical
charts on the same Lagrangian manifold. It is convenient to
introduce a similar factor in the amplitude of the general Fourier
integral \eqref{FIO}. Thus, let $\Phi(x,\theta)$, $(x,\theta)\in
V\subset\RR_x^{n}\times\RR_{\theta}^m$, be a nondegenerate phase
function, and let $d\mu$ be a measure on $\La_\Phi$. Using the
identification \eqref{e507a} of $\La_\Phi$ and $C_\Phi$, we can
treat $d\mu$ as a differential $n$-form on $C_\Phi$. Next, let
$\widetilde{d\mu}$ be an arbitrary differential $n$-form defined in
a neighborhood of $C_\Phi$ in $V$ such that
\begin{equation}\label{e-emb}
    \imath^*(\widetilde{d\mu})=d\mu,\qquad \text{where $\imath\colon C_\Phi\lra V$
    is the embedding.}
\end{equation}
The product $\widetilde{d\mu}
\wedge(-d\Phi_\theta)\overset{\operatorname{def}} =\widetilde{d\mu}
\wedge d(-\Phi_{\theta_1})\wedge\dotsm\wedge d(-\Phi_{\theta_m})$
is a differential form of maximum degree $n+m$, and hence
\begin{equation*}
\widetilde{d\mu} \wedge(-d\Phi_\theta)=F(x,\theta)\, dx\wedge
d\theta\equiv F(x,\theta)\, dx_1\wedge\dotsm\wedge dx_n\wedge
d\theta_1\wedge\dotsm\wedge
    d\theta_m
\end{equation*}
for some function $F(x,\theta)$. We write
\begin{equation}\label{e515}
    F(x,\theta)\equiv F[\Phi,d\mu](x,\theta)=\frac{\widetilde{d\mu} \wedge(-d\Phi_\theta)}
    {dx\wedge d\theta}.
\end{equation}
The restriction of $F[\Phi,d\mu](x,\theta)$ to $C_\Phi$ is
independent of the choice of the form $\wt{d\mu}$ satisfying
\eqref{e-emb}. Indeed, let $\xi_1,\dotsc,\xi_{n+m}\in\RR^{n+m}$ be
linearly independent vectors such that $\xi_1,\dotsc,\xi_n$ form a
basis in the tangent space to $C_\Phi$ at some point
$(x_*,\theta_*)$. Then
\begin{equation*}
    (\widetilde{d\mu}
    \wedge(-d\Phi_\theta))(\xi_1,\dotsc,\xi_{n+m})=
    -d\mu(\xi_1,\dotsc,\xi_n)d\Phi_\theta (\xi_{n+1},\dotsc,\xi_{n+m})
\end{equation*}
depends only on $d\mu$. Moreover, it is nonzero, because $d\mu$ is
nondegenerate on the tangent space to $C_\Phi$ and $d\Phi_\theta$
is nondegenerate in the transversal directions. From now on, we
consider Fourier integrals with amplitude
\begin{equation}\label{ampli}
    a=\ph\sqrt{F[\Phi,d\mu]},
\end{equation}
where $\ph$ is some function on $C_\Phi$ (or, equivalently, on
$\La_\Phi$).

\subsection{Fourier integral as the canonical operator}

Let $\Phi(x,\theta)$, $(x,\theta)\in
V\subset\RR_x^{n}\times\RR_{\theta}^m$, be a nondegenerate phase
function, and let $d\mu$ be a measure on $\La_\Phi$. We claim that
the Fourier integral $\cI[\Phi,\ph\sqrt{F[\Phi,d\mu]}]$ is none
other than the canonical operator on $\La_\Phi$ applied to the
function $\ph$,
\begin{equation}\label{ravenstvo}
    \cI[\Phi,\ph\sqrt{F[\Phi,d\mu]}]=K_{(\La_\Phi,d\mu)}^{1/h}\ph+O(h),
\end{equation}
for an appropriate choice of the action and arguments of Jacobians
on $\La_\Phi$. Without loss of generality, we assume that
$\La_\Phi$ is covered by a single canonical chart $(U,I)$,
$U=\La_\Phi$.
\begin{theorem}\label{atheorem1}
Let the canonical operator $K_{(\La_\Phi,d\mu)}^{1/h}$ be defined
by formula~\eqref{KO-1}, where the eikonal $\tau_{(U,I)}$ coincides
with the eikonal $\tau$ defined in Proposition~\rom{\ref{apropo2}}
and the index $m_{(U,I)}$ is chosen according to the rule
\begin{equation}\label{argJ}
    m_{(U,I)}=-\frac1\pi\arg F[\Phi,d\mu]
    -\si_-\left[\begin{pmatrix}
      -\Phi_{\theta\theta} & -\Phi_{\theta x_\I} \\
      -\Phi_{x_\I\theta} & -\Phi_{x_\I x_\I} \\
    \end{pmatrix}\right]+\abs{\I},
\end{equation}
$\si_-(A)$ being the number of negative eigenvalues of a symmetric
matrix $A$. Then relation~\eqref{ravenstvo} holds.
\end{theorem}
\begin{corollary}
Two phase functions are equivalent \rom(i.e., define the same space
of Fourier integrals\rom) if and only if the corresponding
Lagrangian manifolds are the same.
\end{corollary}
Indeed, Theorem~\ref{atheorem1} reduces an arbitrary Fourier
integral to the canonical operator on the corresponding Lagrangian
manifold.
\begin{proof}[Proof \rm of Theorem~\ref{atheorem1}]
Let us apply the $1/h$-Fourier transform from the variables $x_\I$
to the variables $p_\I$ (see Sec.~\ref{s6}) to
Eq.~\eqref{ravenstvo}. We see that it suffices to prove that
\begin{multline}\label{shlyapa}
    \frac{e^{i\pi(m-\abs{\I})/4}}{(2\pi h)^{(m+\abs{\I})/2}}
    \int e^{\frac ih[\Phi(x,\theta)-p_\I
    x_\I]}\ph(x,\theta)\sqrt{F[\Phi,d\mu](x,\theta)}\,d\theta\,dx_\I
    \\= e^{\frac ih[\tau(\a(x_I,p_\I))-p_\I X_\I(x_I,p_\I)]}
    \frac{\ph(\a(x_I,p_\I))}{\sqrt{\cJ_I(\a(x_I,p_\I))}} +O(h).
\end{multline}
To this end, we use the stationary phase method (see
Theorem~\ref{spm} below). The stationary point equations for the
phase function
\begin{equation*}
    \Psi(x,\theta,p_\I)=\Phi(x,\theta)-p_\I x_\I
\end{equation*}
of the integral on the left-hand side in~\eqref{shlyapa} read
\begin{equation}\label{stpeq}
    \Phi_\theta(x,\theta)=0,\qquad \Phi_{x_\I}(x,\theta)-p_\I=0.
\end{equation}
In particular, if a point $(x_\I,\theta)$ is a stationary point of
the integral for given $(x_I,p_\I)$, then $(x,\theta)\in\La_\Phi$.
Let us compute $F[\Phi,d\mu](x,\theta)$ at the stationary points.
We can take
\begin{equation*}
    \wt{d\mu}=\mu_I(x_I,\Phi_{x_\I}(x,\theta))\, dx_I\wedge
    d(\Phi_{x_\I}(x,\theta)), \quad
    \mu_I(x_I,p_\I)=\frac1{\cJ_I(\a(x_I,p_\I))},
\end{equation*}
and so \eqref{e515} gives (all computations are carried out for
$(x,\theta)\in C_\Phi$)
\begin{equation}\label{crucial}
\begin{split}
    F[\Phi,d\mu]&=\mu_I \frac{dx_I\wedge
    d\Phi_{x_\I}\wedge(-d\Phi_\theta)}
    {dx\wedge d\theta}=(-1)^{\abs{\I}}\mu_I \frac{dx_I\wedge
    d'(-\Phi_{x_\I})\wedge(-d'\Phi_\theta)}
    {dx_I\wedge dx_\I\wedge d\theta}\\
    &=(-1)^{\abs{\I}}\mu_I \frac{d'(-\Phi_{x_\I})\wedge(-d'\Phi_\theta)}
    {dx_\I\wedge d\theta}
    =(-1)^{\abs{\I}}\mu_I \det\begin{pmatrix}
      -\Phi_{\theta\theta} & -\Phi_{\theta x_\I} \\
      -\Phi_{x_\I\theta} & -\Phi_{x_\I x_\I} \\
    \end{pmatrix}.
\end{split}
\end{equation}
Here $d'$ is the differential with respect to the variables
$(x_\I,\theta)$, the variables $x_I$ being treated as parameters.
Since $F[\Phi,d\mu]\ne0$ on $C_\Phi$, we see that the determinant
is nonzero and the stationary point equations~\eqref{stpeq} are
nondegenerate, so that their solution is given by smooth functions
$x_\I=x_\I(x_I,p_\I)$, $\theta=\theta(x_I,p_\I)$. Since the point
$(x_I,x_\I(x_I,p_\I),\theta(x_I,p_\I))$ lies in $C_\Phi$, it
follows that
$$
(x_I,x_\I(x_I,p_\I),\Phi_x(x_I,x_\I(x_I,p_\I),\theta(x_I,p_\I)))
\in\La_\Phi.
$$
By the second equation in~\eqref{stpeq}, we can
replace $\Phi_{x_\I}$ by $p_\I$ here and obtain
\begin{equation*}
    (x_I,x_\I(x_I,p_\I),\Phi_{x_I}(x_I,x_\I(x_I,p_\I),\theta(x_I,p_\I)),p_\I)
\in\La_\Phi,
\end{equation*}
or, in view of Eqs.~\eqref{loc-coor0},
\begin{equation*}
    x_\I (x_I,p_\I)=X_\I(x_I,p_\I),\qquad
    \Phi_{x_I}(x_I,x_\I(x_I,p_\I),\theta(x_I,p_\I))=
    P_I(x_I,p_\I).
\end{equation*}
Thus, the point $(x_I,x_\I(x_I,p_\I),\theta(x_I,p_\I))\in C_\Phi$
corresponds via the mapping $j_\Phi$ to the point
$(x_I,X_\I(x_I,p_\I),P_I(x_I,p_\I),p_\I)\in\La_\Phi$, and
accordingly (see Proposition~\ref{apropo2}) we have
$\Phi(x_I,x_\I(x_I,p_\I),\theta(x_I,p_\I))=\tau(\a(x_I,p_\I))$. The
phase function at the stationary point is
\begin{multline*}
    \Psi(x_I,x_\I(x_I,p_\I),\theta(x_I,p_\I),p_\I)
    =\Phi(x_I,x_\I(x_I,p_\I),\theta(x_I,p_\I))-p_\I X_\I(x_I,p_\I)
    \\=\tau(\a(x_I,p_\I))-p_\I X_\I(x_I,p_\I).
\end{multline*}
Now we apply Theorem~\ref{spm} and obtain
\begin{multline}\label{shlyapa1}
    \frac{e^{i\pi(m-\abs{\I})/4}}{(2\pi h)^{(m+\abs{\I})/2}}
    \int e^{\frac ih(\Phi(x,\theta)-p_\I
    x_\I)}\ph(x,\theta)\sqrt{F[\Phi,d\mu](x,\theta)}\,d\theta\,dx_\I
    \\= e^{-i\pi\abs{\I}/2} e^{\frac
    ih[\tau(\a(x_I,p_\I))-p_\I X_\I(x_I,p_\I)]}
    \left[\ph(x,\theta)\sqrt{\frac{F[\Phi,d\mu](x,\theta)}{\det\begin{pmatrix}
      -\Phi_{\theta\theta} & -\Phi_{\theta x_\I} \\
      -\Phi_{x_\I\theta} & -\Phi_{x_\I x_\I} \\
    \end{pmatrix}}}\right]+O(h),
\end{multline}
where the expression in square brackets is taken at the stationary
point and the argument of the determinant is chosen as indicated in
Theorem~\ref{spm}. Now we take into account \eqref{argJ},
\eqref{crucial}, and the fact that $\mu_I(x_I,p_\I)=
\cJ(\a(x_I,p_\I))^{-1}$ and arrive at \eqref{shlyapa}. The proof of
the theorem is complete.
\end{proof}

\section{Auxiliary information}\label{s6}

\subsection{Notation}

All vectors are understood as column vectors. If $\xi$ and $\eta$
are $n$-vectors, then we write $\langle \xi,\eta\rangle$ for the
bilinear form $\langle \xi,\eta\rangle=\sum_{j=1}^n \xi_j\eta_j.$
Sometimes, however, we just write $\xi\eta$ instead. Partial
derivatives are denoted by subscripts; for example,
$\Phi_x=\pa\Phi/\pa x$.

If $I$ is a subset of $\{1,\dotsc,n\}$, then by $\I$ we denote the
complementary subset $\I=\{1,\dotsc,n\}\setminus I$. By $\abs{\I}$
we denote the number of elements in $\I$. The $\abs{I}$-vector with
components $x_j$, $j\in I$, is denoted by $x_I$. The product
$dx_I\wedge dp_\I$ is understood as the exterior product of all
differentials $dx_j$, $j\in I$, and $dp_j$, $j\in\I$, all factors
being arranged in ascending order of subscripts from $1$ to $n$.
Similar intuitively clear notation is used as well. Next, $A_{I\I}$
is an $\abs{I}\times\abs{\I}$ matrix with entries $A_{jk}$, $j\in
I$, $k\in\I$. For example, if $I=\{1,3\}$ and $\I=\{2,4\}$, then
\begin{gather*}
    dx_I=dx_1\wedge dx_3,\quad
    dp_\I=dp_2\wedge dp_4,\quad
    dx_I\wedge dp_\I=dx_1\wedge dp_2\wedge dx_3\wedge dp_4,\\
    (x_I,p_\I)=(x_1,p_2,x_3,p_4),\qquad
    \pd{S_I}{x_I\pa p_\I}=\begin{pmatrix}
      \pd{S_I}{x_1\pa p_2} & \pd{S_I}{x_1\pa p_4} \\
      \pd{S_I}{x_3\pa p_2} & \pd{S_I}{x_3\pa p_4} \\
    \end{pmatrix}.
\end{gather*}

\subsection{Fourier transform}

Recall that the $1/h$-Fourier transform from the variables $x_\I$
to the variables $p_\I$ and the inverse transform are defined as
\begin{align}\label{FT0}
    [\cF_{x_\I\to p_\I}^{1/h}u](x_I,p_\I)
    &=\frac{e^{-i\pi\abs{\I}/4}}{(2\pi h)^{\abs{\I}/2}}
    \int e^{-\frac ih p_\I x_\I}u(x)\,dx_\I,
\\ \label{FT1}
   [\bar\cF_{p_\I\to x_\I}^{1/h}v](x)
    &=\frac{e^{i\pi\abs{\I}/4}}{(2\pi h)^{\abs{\I}/2}}
    \int e^{\frac ih p_\I x_\I}v(x_I,p_\I)\,dp_\I.
\end{align}

\subsection{Asymptotics of oscillatory integrals} Here we reproduce the statement of the theorem on
the stationary phase method used in the preceding subsection.

\begin{theorem}[e.g., see~{\cite{Fdr1,MSS1}}]\label{spm}
Let the phase function $\Phi(x,\theta)$ have a unique stationary
point $\theta=\theta(x)$ on the support of the amplitude, and
assume that this stationary point is nondegenerate,
$\det\Phi_{\theta\theta}(x,\theta(x))\ne0$. Then the integral
\eqref{FIO} has the asymptotics
\begin{equation}\label{sfm1}
    I[\Phi,a](x)
    =\frac{e^{\frac
    ih\Phi(x,\theta(x))}a(x,\theta(x))}
    {\sqrt{\det(-\Phi_{\theta\theta}(x,\theta(x)))}}+O(h),
\end{equation}
where the branch of the square root in the denominator is chosen
according to the rule
\begin{equation}\label{sfm2}
   \arg\det(-\Phi_{\theta\theta}(x,\theta(x)))
   =-\pi\si_-(-\Phi_{\theta\theta}(x,\theta(x))).
\end{equation}
Here $\si_-(A)$ is the negative index of inertia \rom(the number of
negative eigenvalues\rom) of a self-adjoint matrix~$A$.
\end{theorem}

\end{document}